\documentclass[pdfa,cleveref,autoref]{lipics-v2021}
\pdfoutput=1

\nolinenumbers

\usepackage{graphicx} 
\usepackage{bbold}

\usepackage{algorithm}
\usepackage{algpseudocodex}
\usepackage{booktabs, caption, longtable}

\usepackage{todonotes}
\setuptodonotes{backgroundcolor=blue!10,textcolor=purple}
\todostyle{JM}{author=JM,textcolor=red!50!black}
\todostyle{MVJ}{author=MVJ,textcolor=blue!50!black}

\usepackage{tikz}
\usepackage{pgfplots}
\pgfplotsset{compat=1.18}

\title{Generating Cofaces in Vietoris--Rips Filtration Order}
\author{Ulrich Bauer}{Technical University of M\"unchen}{bauer@tum.de}{}{}
\author{Jordan Matuszewski}{City University of New York -- the Graduate Center}{jmatuszewski@gradcenter.cuny.edu}{}{This work was supported by a grant from the Simons Foundation (961833, MVJ).}
\author{Mikael Vejdemo-Johansson}{City University of New York -- College of Staten Island \and College of Staten Island}{mvj@math.csi.cuny.edu}{0000-0001-6322-7542}{This work was supported by a grant from the Simons Foundation (961833, MVJ).}
\authorrunning{U.~Bauer \and J.~Matuszewski \and M.~Vejdemo-Johansson}

\Copyright{Ulrich Bauer, Jordan Matuszewski, Mikael Vejdemo-Johansson}

\date{September 2024}

\keywords{Vietoris-Rips complex, cofaces, cofacets}
\ccsdesc[500]{Theory of computation~Computational geometry}

\begin{document}

\maketitle

\begin{abstract}
We propose novel algorithms for generating cofaces in a Vietoris-Rips complex.
Cofaces -- simplices that contain a given simplex -- have multiple important uses in generating and using a Vietoris-Rips filtration: both in creating the coboundary matrix for computing persistent cohomology, and for generating the ordered sequence of simplices in the first place.
Traditionally, most methods have generated simplices first, and then sorted them in filtration order after the generation step.
In this paper, we propose fast algorithms  for generating the sequence of simplices by generating cofaces of a given simplex with the same diameter, which by construction produces simplices in filtration order, and for generating additional cofaces in filtration order using sorted neighborhood lists in order to generate coboundaries directly in filtration order.
\end{abstract}

\section{Introduction}

In Topological Data Analysis, and especially in persistent (co)homology, everything starts with a filtered simplicial complex.
Out of many options, by far the most common is the Vietoris-Rips complex: the simplicial complex that to a filtration parameter $\alpha$ assigns the flag complex of the $\alpha$-neighborhood graph.
Construction of this complex traditionally follows the Incremental construction by Afra Zomorodian in \cite{zomorodian_fast_2010}.

As demonstrated in \cite{Chen2011Persistent,bauer_clear_2014,bauer_ripser_2021}, considerable improvements can achieved by moving up through the dimensions, instead of going through parameter values the way that Zomorodian's Incremental algorithm does, especially since it allows for the \emph{clearing} operation \cite{Chen2011Persistent}.
Thus, in the popular Ripser package \cite{bauer_ripser_2021}, after special handling of the first two dimensions (vertices and edges), the software relies on generating all $k$-simplices in a Vietoris-Rips complex by going through $(k-1)$-simplices and generating cofacets (that is, cofaces of codimension $1$) by adding vertices.
In the design of Ripser, this cofacet generation fulfills two separate tasks: on the one hand, it is used to create all simplices in the next pass over the complex to compute the barcode for the next homological degree; on the other hand, it is used to generate all cofacets of a particular simplex when creating a new column of the coboundary matrix.

In this paper, we will introduce new algorithms for generating cofacets of a given simplex.
One of the algorithms, largely discussed in \cref{sec:simplex-stream}, is focused on producing a sequence of simplices in Vietoris-Rips filtration order (more precisely, a lexicographic refinement to a simplexwise filtration) without any duplicates.
The other algorithm, discussed in \cref{sec:coboundary-generation}, is focused on computing the coboundary of a specific simplex, with all cofacets listed in filtration order.

\section{Method}
\label{sec:method}

Most constructions in our paper stem from one particular notion, present in Ripser as \emph{top cofacets}:

\begin{remark}
\label{rem:facet-vertex-pair}
A cofacet in a simplicial complex is completely described by one of its facets, paired with a vertex.
If we create a notion of a particular facet-vertex pair being \emph{responsible} for generating a specific new cofacet, with a bijection between the $k$-simplices and a subset of these pairs, we can uniquely enumerate all cofacets by finding all such pairs.
\end{remark}

In Ripser, this enumeration approach is present as the \emph{top cofacet enumeration}, where the \emph{lexicographically first} facet is responsible for generating a cofacet: for a simplex $[v_d,\dots,v_0]$, the first facet would be $[v_{d-1},\dots,v_0]$ and thus the pair $(v_d,[v_{d-1},\dots,v_0])$ would be responsible for generating the cofacet $[v_d,\dots,v_0]$.
With this definition, it is easy to recognize valid pairs: these are the ones where the vertex paired with a simplex comes before all vertices of the simplex in the vertex order.
And it is also easy to recognize what the possible cofacets of a particular simplex can be: since $v_d>v_{d-1}$, we can just iterate through all choices of additional vertex $v > v_{d-1}$.
This produces the cofacets in reverse colexicographic ordering, which works well enough for the purposes in Ripser.

Our work builds on the following observation:
\begin{remark}\label{rem:longest-edge}
Any simplex with Vietoris-Rips filtration value $\alpha$ has at least one edge of length exactly $\alpha$, and no longer edges.
\end{remark}

Ripser's approach generates all simplices in a complex in reverse colexicographic order, which does not have to be the same order as the one induced by the Vietoris-Rips filtration values.
The problem can easily be seen: an additional vertex $v$ might be at a distance greater than $\alpha$ from one of the vertices of a simplex, which would mean that the filtration value increases as we go from the simplex $[v_{d-1},\dots,v_0]$ to its cofacet.

Our ambition in this paper is to generate the Vietoris-Rips complex by cofacets, but already sorted by filtration values.
If we have a list of all the $(d-1)$-dimensional simplices ordered by filtration value, and we can find cofacets \emph{without increasing filtration values}, we will get a sorted list; if we also make sure that no $d$-simplex gets generated more than once, we can generate the simplex stream we seek.
If we never increase the filtration values, we can start with an edge of length $\alpha$, and gradually generate all cofaces of diameter $\alpha$ of this edge by generating cofacets and cofacets of cofacets.
But doing this, we no longer have access to the simple rule described above: it could well be that there are vertices $v > v_{d-1}$ that are not suitable (because they would increase the filtration value), and it could be that a new vertex would have to be $v < v_{d-1}$ to construct the bijection we are looking for.
We will discuss our approach in much more detail in \cref{sec:simplex-stream} -- the solution we arrive at is to make the \emph{lexicographically first full diameter} facet responsible for generating a cofacet.

\subsection{Notation and setup}
\label{sec:notation-setup}

We assume we have some finite metric space $(V,d)$, and want the Vietoris-Rips simplex stream and coboundary operator for of this set of points.
We write explicit simplices with square brackets containing vertices in decreasing order.

We will use \emph{coface} for any simplex $\tau$ containing a given simplex $\sigma$ and \emph{cofacet} when $|\tau| = |\sigma|+1$. When convenient, we use $(\sigma + v)$ in lieu of $\tau$. In the same fashion,  $(\tau - v)$ may be used in place of $\sigma$ at times.

Similarly, a \emph{face} $\sigma$ is contained in a simplex $\tau$ and it is a \emph{facet} when they differ by exactly one vertex.

\subsection{Simplex Stream generation}
\label{sec:simplex-stream}


Maintaining filtration order of our Vietoris-Rips simplex stream requires that at every parameter value $\alpha$ when a $(d-1)$-simplex $\sigma = [v_d,\dots,v_1]$ appears we can identify the diameter-$\alpha$ $d$-simplices $\tau$ for which $\sigma$ is the lexicographically first facet. 

We get there by classifying the possible \emph{lexicographically first full-diameter facets} of a simplex.

\begin{theorem}\label{thm:three-cases}
    Given a simplex $\tau = [v_d,\dots,v_0]$ of diameter $\alpha$, its lexicographically first full-diameter facet is one of:
    \begin{enumerate}
        \item $[v_{d-1},\dots,v_0]$ if at least one of the edges in the facet has length $\alpha$;
        \item $[v_d,v_{d-2},\dots,v_0]$ if every full length edge contains $v_d$, and at least one of the edges in this facet has length $\alpha$;
        \item $[v_d,v_{d-1},v_{d-3},\dots,v_0]$ if $v_dv_{d-1}$ is the only full length edge in $\tau$.
    \end{enumerate}
\end{theorem}
\begin{proof}
    Let $\sigma_i=[v_d,\dots,\hat{v}_i,\dots,v_0]$ be the facet of $\tau$ resulting from removing $v_i$.
    Note that in the lexicographic order, $\sigma_i > \sigma_d$, since $v_d > v_{d-1}$ (recall that simplices are written in decreasing vertex order, see \cref{sec:notation-setup}).
    The only way for $\sigma_d$ to not be the lexicographically first full-diameter facet is if $\sigma_d$ is not full-diameter, so that all $d(v_i,v_j)<\alpha$ for $i,j < d$.

    Now, suppose that $\sigma_d$ is not the lexicographically first full-diameter facet.
    Then none of the edges $v_iv_j$ where $i,j < d$ have length $\alpha$.
    Hence, since $\tau$ \emph{has} diameter $\alpha$, at least one of the remaining edges has to have length $\alpha$.
    These remaining edges are all on the shape $v_dv_i$, for $i<d$.

    Among the remaining facets (disqualifying $\sigma_d$ by our assumption), the lexicographically first one is $\sigma_{d-1}$, since $v_{d-1} > v_{d-2}$ (and we would be comparing $[v_d,v_{d-1},\dots,\hat{v}_i,\dots,v_0]$ to $[v_d,v_{d-2},\dots,v_0]$).
    If $\sigma_{d-1}$ has full diameter, it is therefore the lexicographically first facet of full diameter.
    This implies that at least one of the edges $v_dv_i$ with $i<(d-1)$ is full length.

    The only way for $\sigma_{d-1}$ to not be the lexicographically first full-diameter face is if $v_dv_{d-1}$ is the only full length edge in $\tau$.

    Finally, assume that neither $\sigma_d$ nor $\sigma_{d-1}$ is the lexicographically first full-diameter facet.
    Then $v_dv_{d-1}$ is the only edge in $\tau$ with full length.
    Among remaining facets, the lexicographically first facet is $\sigma_{d-2}$, because $v_{d-2} > v_{d-3}$.
    Since $d(v_d,v_{d-1})=\alpha$, the simplex $\sigma_{d-2}=[v_d,v_{d-1},v_{d-3},\dots,v_0]$ has diameter $\alpha$.
\end{proof}

This theorem gives us three different cases to consider when trying to generate cofacets of a given simplex $\sigma=[v_d,\dots,v_1]$ of diameter $\alpha$.

To understand the three cases, consider a $d$-simplex $\tau = [t_d,\dots,t_0]$. Following Bauer's approach closely, we want $\tau$ to be represented by the lexicographically first facet \emph{of full diameter}. 

If $[t_{d-1},\dots,t_0]$ has diameter $\alpha$, we do not need to do anything further - and any $t > t_{d-1}$ with sufficiently short edges will be an appropriate vertex to adjoin. This produces case 1 above.

If, however, in $\tau$ all the full length edges connect to $t_d$, then $[t_{d-1},\dots,t_0]$ has a lower diameter than $\alpha$. To keep the filtration order, this means that we cannot represent such a $\tau$ as a cofacet of $[t_{d-1},\dots,t_0]$. Then the lexicographically first facet of full diameter would be $[t_d,t_{d-2}\dots,t_0]$, where $t_{d-1}$ gets removed. Looking towards the generation problem, this observation translates to case 2 below -- if the new vertex is $v_{d-1} > v > v_{d-2}$, and all full length edges are of the form $[v_{d-1},v_i]$ or $[v_{d-1},v]$ (and all other new edges are shorter than $\alpha$), then we get an appropriate new vertex to adjoin.

This second construction, however, fails if in $\tau$ there is exactly one full length edge, namely $[t_d,t_{d-1}]$. If this is the case, dropping either $t_{d}$ or $t_{d-1}$ gives us a facet with lower diameter than $\alpha$. So to get the lexicographically first facet with full diameter, we would then need to drop $t_{d-2}$. This produces case 3, where new vertices to consider are precisely $v$ such that $v_{d-2}>v>v_{d-3}$ and all new distances are strictly shorter than $\alpha$.

For an illustrative example of how this works in a simplex stream generation, going from $(d-1)$-simplices to their $d$-dimensional cofacets, consider the following figure:

\begin{tikzpicture}
\filldraw[black] (0,0) circle (2pt) node[anchor=north east] {$v_2$};
\filldraw[black] (1,0) circle (2pt) node[anchor=north west] {$v_1$};
\filldraw[black] (0,1) circle (2pt) node[anchor=south east] {$v$};
\filldraw[black] (3,1) circle (2pt) node[anchor=south west] {$v_3$};
\draw[black,very thick] (0,0) -- (1,0);
\draw[black] (0,0) -- (0,1);
\draw[black] (0,1) -- (1,0);
\draw[black,very thick] (0,0) -- (3,1);
\draw[black] (0,1) -- (3,1);
\draw[black,very thick] (1,0) -- (3,1);
\end{tikzpicture}

If we are looking for possible vertices $v$ to add, we would want to iterate through:

\paragraph*{Case 1: $v > v_d$}

Similar to Bauer's approach for responsible-pair generation we consider as candidates all vertices $v > v_d$.
Such a vertex produces a cofacet for which $\sigma$ is the lexicographically first \emph{full diameter} facet of $\sigma + v$, as long as $\sigma+v$ does not have a larger diameter than $\alpha$.
Establishing this corresponds to checking that all distances $d(v,v_i)\leq\alpha$.

\paragraph*{Case 2: $v_d > v > v_{d-1}$ }

This case we only need to evaluate if every full length edge in $\sigma$ connects to $v_d$.
In this case, any vertex $v$ between $v_d$ and $v_{d-1}$ could be a candidate producing a cofacet.
We need to check that all full length edges in $\sigma$ connect to $v_d$, and also that $[v,v_d]$ is the only new edge that could have length $\alpha$: all other $d(v,v_j) < \alpha$
These conditions ensure that for the produced cofacet $\tau$, all full-length edges connect to the first vertex.

\paragraph*{Case 3: $v_{d-1} > v > v_{d-2}$:}

This last case corresponds to the 3rd case in \cref{thm:three-cases}.
This case only needs to be evaluated if $v_dv_{d-1}$ is the only full length edge in $\sigma$.
For such $\sigma$, any vertex $v$ between $v_{d-1}$ and $v_{d-2}$ could be a candidate producing a cofacet.
We need to check for any candidate vertex that all edges from $v$ to any vertex of $\sigma$ is strictly shorter than $\alpha$.

A different perspective may serve to develop the intuition: working in reverse from $\tau$, the removal of $v_d$ and $v_{d-1}$ results in facets ($\tau - v_d$) and ($\tau - v_{d-1}$) respectively, neither of which are full diameter. As such, we select the next least vertex whose removal results in the full diameter facet $\tau - v$. 


In \cref{app:worked-example}, we include a detailed worked example that illustrates how these candidate vertices emerge and can be checked.

\subsection{Apparent Pairs}
\label{sec:apparent-pairs}

One of the fundamental breakthroughs in the Ripser approach was the early recognition of \emph{apparent pairs} \cite{bauer2021gromov}, that generate a 0-length cohomology interval without much influence on other parts of the computation, and therefore can be dropped from the computation without further processing.

We recall the definition of an apparent pair:

\begin{definition}
    A pair of simplices $\sigma \subset \tau$ is an \emph{apparent pair} if $\sigma$ is the lexicographically first same-diameter facet of $\tau$ and $\tau$ is the lexicographically last same-diameter cofacet of $\sigma$.
\end{definition}

Bauer and Edelsbrunner \cite{bauer2017morse} provide a discrete morse gradient field built from apparent pairs showing that they can be removed without influencing the resulting (co)homology computation.

The three cofacet generation cases we identify in Section~\ref{sec:method} contribute to an approach for fast identification of an apparent pairing while generating a simplex stream:

If we find ourselves in case 1., with $\tau = [v,v_{d-1},\dots,v_0]$ with diameter $\alpha$, then clearly $\sigma=[v_{d-1},\dots,v_0]$ is the lexicographically first full diameter facet. $\tau$ is in an apparent pair exactly when $v$ is the last vertex within an $\alpha$-neighborhood of $\sigma$.

In case 2., with $\tau = [v_{d-1},v,v_{d-2},\dots,v_0]$ with diameter $\alpha$ and all full-length edges connecting to $v_{d-1}$, and at least one of these full-length edges connecting to some $v_j, j\neq d-1$, this is in an apparent pair with $\sigma=[v_{d-1},v_{d-2},\dots,v_0]$ precisely if there are no vertices later than $v$ in an $\alpha$-neighborhood of $\sigma$.

In case 3., with $\tau = [v_{d-1},v_{d-2},v,v_{d-3},\dots,v_0]$ with diameter $\alpha$ and exactly one full-length edge, namely $v_{d-1}v_{d-2}$, this is an apparent pair with $\sigma = [v_{d-1},v_{d-2},v_{d-3},\dots,v_0]$ precisely if there are no vertices later than $v$ in an $\alpha$-neighborhood of $\sigma$.

That these are the only cases needed to consider for apparent pairs follows from the following:

\begin{proposition}
\label{thm:cases-for-cofacet}
    A simplex $\tau=[v_d,\dots,v_0]$, with vertices written in increasing order, can be an apparent pair cofacet of exactly one of $[v_{d-1},\dots,v_0]$, $[v_d,v_{d-2},\dots,v_0]$ or $[v_d,v_{d-1},v_{d-3},\dots,v_0]$.
\end{proposition}
\begin{proof}
    The lexicographically first same-diameter facet of $\tau$ is $[v_{d-1},\dots,v_0]$, except if the exclusion of $v_{d-1}$ lowers the diameter.

    This exclusion lowers the diameter precisely if all full-length edges in $\tau$ connect to $v_d$. If there is at least one such edge that does not connect to $v_{d-1}$, the lexicographically first same-diameter facet of $\tau$ is $[v_d,v_{d-2},\dots,v_0]$.

    If the only full-length edge in $\tau$ is $v_dv_{d-1}$, then the lexicographically first same-diameter facet is $[v_d,v_{d-1},v_{d-3},\dots,v_0]$.

    For any other facet $[v_d,\dots,\widehat{v_j},\dots,v_0]$, there is an earlier facet given by one of these three, which has the same diameter by construction and comes earlier since $v_d > v_{d-1}$ (for the first case), $v_{d-1} > v_{d-2}$ (for the second case) or $v_{d-2} > v_{d-3}$ (for the third case).
\end{proof}

\subsection{In-order Simplex Streams}
\label{sec:in-order-simplex-streams}

\cref{thm:three-cases} further enables us to create a simplex stream generation algorithm that avoids storing the entire previous layer of simplices in memory to create the next.
Instead of iterating through all $(d-1)$-simplices to create all $d$-simplices, we can iterate through the edges and for each edge create cofacets and cofacets of cofacets until we hit our target dimension.
In this way, apparent pair simplices are not stored at all -- they show up transiently when generating their cofacets for higher dimensional simplex enumeration, but reduce the overall memory load.

We have implemented this approach using a stack of enumerator objects, each walking through the top cofacets of its assigned simplex.
Once the top enumerator is empty, it gets popped off the stack, and the next enumerator is advanced to find a simplex to assign to a new enumerator.

Writing the peak memory usage for the simplex streams in Ripser as 
\[
O\left(|\Sigma_{d-1}|\cdot\texttt{sizeof(simplex)} + |\Sigma_d\setminus \{\text{apparent pair simplices}\}|\cdot\overline{|\delta \sigma|}\cdot\texttt{sizeof(simplex)}\right)
\]
where $\Sigma_{d-1}, \Sigma_d$ are the simplices of dimension $d-1$ and $d$ respectively, and $\overline{|\delta\sigma|}$ is the mean number of non-zero entries in a coboundary of a $(d-1)$-simplex, this change would instead create a peak memory usage of 
\[
O\left(d\cdot\texttt{sizeof(enumerator)} + |\Sigma_d\setminus \{\text{apparent pair simplices}\}|\cdot\overline{|\delta \sigma|}\cdot\texttt{sizeof(simplex)}\right)
\]

Each enumerator in this stack holds the simplex $\sigma=[v_d,\dots,v_0]$ of diameter $\alpha$ for which it generates cofacets, as well as the intersection $N$ of all $\alpha$-neighborhoods of its vertices.

\begin{remark}\label{rmk:cofacet-enumeration}
Enumerating the cofacets then proceeds by:

\begin{itemize}
    \item For each $v\in N$, such that $v>v_d$, emit the corresponding simplex $[v,v_d,\dots,v_0]$.
    \item If each full-length edge in $\sigma$ is incident with $v_d$, proceed with all $v\in N$ such that $v_d > v > v_{d-1}$. For each of these, if all vertices of $\sigma$ except possibly $v_d$ are strictly shorter than $\alpha$, emit the corresponding simplex $[v_d,v,v_{d-1},\dots,v_0]$.
    \item If $[v_d,v_{d-1}]$ is the only full-length edge in $\sigma$, proceed with $v\in N$ such that $v_{d-1} > v > v_{d-2}$ (if $\sigma$ is an edge, interpret $v_{d-2}=-\infty$, and include all vertices in the consideration). If all distances from $v$ to a vertex in $\sigma$ are strictly shorter than $\alpha$, emit the corresponding simplex $[v_d,v_{d-1},v,v_{d-2}]$.
\end{itemize}
\end{remark}

\subsubsection{Representing full-length edges}
\label{sec:representing-full-length-edges}

Since the criteria in \cref{thm:three-cases} concern the edges of full length, it would be useful to have a lean way of accessing information about these full length edges without iterating through $O(d^2)$ distances every time.
We suggest tracking the lexicographically first edge (written with vertices in descending order) among the full-length edges.

\begin{proposition}\label{thm:full-length-edge}
    Let $[s,t]$ be the lexicographically first full-length edge of a simplex $[v_d,\dots,v_0]$.
    Then
    \begin{enumerate}
        \item\label{case:case-2} If $s=v_d$, then all full-length edges are incident to $v_d$.
        \item\label{case:case-3} If $s=v_d$ and $t=v_{d-1}$, then this is a unique full-length edge.
        \item\label{case:same-longedge} $[s,t]$ is the lexicographically first full-length edge of every cofacet $\tau\supseteq\sigma$ generated according to the method in \cref{rmk:cofacet-enumeration}.
    \end{enumerate}
\end{proposition}

Note that by \cref{case:case-2} and \cref{case:case-3} in \cref{thm:full-length-edge}, we can determine whether to consider any $v_d>v$ at all, and what the lowest index vertex to consider is directly from the pair of simplex and first longest edge.
\cref{case:same-longedge} tells us we can create a cofacet simplex object that inherits the first full-length edge from $\sigma$ without any further distance calculations.

\begin{proof}
    For \cref{case:case-2}, note that if $[a,b]$ is a full-length edge in $\sigma$, with $a>b$ and $a\neq v_d$, then $[a,b]$ is lexicographically earlier than any $[v_d,t]$.
    Hence, if the largest vertex of the first full-length edge is $v_d$, then all full-length edges are incident to $v_d$.

    For \cref{case:case-3}, we observe that if $[v_d,t]$ is full-length and $t < v_{d-1}$, then $[v_d,t]$ is lexicographically earlier than $[v_d,v_{d-1}]$.

    Finally, for \cref{case:same-longedge}, if $v > v_d$, then all $[v,v_j]$ come later than the existing $[s,t]$. If instead $v_d > v > v_{d-1}$, then there is already an edge $[v_d,v_j]$ of full length. Since $v>v_{d-1}\geq v_j$, the edge $[v_d,v_j]$ is earlier than $[v_d,v]$. No other $[v,v_i]$ are allowed to be full length. Finally, if $v_{d-1}>v>v_{d-2}$, none of the new edges are allowed to be full length.
\end{proof}

\subsection{Coboundary generation}
\label{sec:coboundary-generation}

We can use a similar approach to generate the collection of \emph{all} cofacets of a given face, in filtration order. Each cofacet of a simplex $\sigma$ is determined by a single vertex to be added, and the challenge here is to traverse these candidate vertices in order of maximal distance to the simplex.

To fix notation let us assume that we have a simplex $\sigma$, with diameter $\alpha$, and we seek to enumerate its cofacets in filtration order.

Drawing on the suggestions in \cite{aggarwal_dory_2024-1}, we will assume that we for each vertex store its neighbors in distance order in a datastructure $\mathcal{N}$, so that $\mathcal{N}(v)$ is a sorted list of neighbors of $v$. We will write $\mathcal{N}(v)(\leq\alpha)$ for the prefix sublist of this list with distances at most $\alpha$.

Our proposed approach for generating a full coboundary is specifically designed to quickly produce and report back an apparent pair, if the simplex is part of one.

Concretely this means that among the filtration value preserving vertices that can be added to $\sigma$, the cofacet $\tau$ is created by picking the last possible vertex. This picked vertex, in $\tau$, needs to be the first possible vertex that can be dropped without reducing diameter.

\subsubsection{Data structure choices}

In our proposed coboundary algorithm, we maintain:
\begin{itemize}
    \item For the underlying metric space, a lookup table (as an array of arrays, or as an associative array containing arrays, for instance) $\mathcal{N}$ that for each vertex $v$ produces a list of all other vertices in order of distance from $v$ annotated with their distance to $v$. See \cref{app:worked-example} for a sketch of what this structure could look like.
    \item For each simplex $\sigma$ we are working with, a collection (as an array, or an associative array containing integers) $\mathcal{P}$ with one integer index for each vertex in $\sigma$ so that $\mathcal{P}(v)$ is tracking how far through the sorted list in $\mathcal{N}(v)$ the algorithm has progressed so far.
    \item An associative array $\mathcal{V}$ that to each vertex $w\not\in\sigma$ that we inspect tracks how many times we have encountered $w$ as a neighbor to some vertex in $\sigma$. If $\mathcal{V}(w)=|\sigma|$, we know that $w$ is in the intersection of all the neighborhoods.
\end{itemize}

\subsubsection{Traverse the $\alpha$-neighborhoods of $\sigma$}

As a first step, in order to report back potential apparent pairs, the algorithm needs to traverse the $\alpha$-neighborhoods of all vertices of $\sigma$ to find the smallest vertex in their intersection.

We do this by linear traversal, increasing the pointer index in $\mathcal{P}(v)$ and storing the vertices observed in a heap data structure $\mathcal{V}_{L1}$, in order to easily recover the observed elements in vertex order.

This initial traversal can then be described as in \cref{alg:first-traversal}.

\begin{algorithm}
\begin{algorithmic}
\For{$v\in \sigma$}
\State $i\gets 0$
\While{$\mathcal{N}(v)(i)(\text{time}) < \alpha$ and $i<|\mathcal{N}(v)|$}
\State $\mathcal{V}_{L1}.\text{enqueue}(\mathcal{N}(v)(i)(\text{vertex}))$
\State $i\gets i+1$
\EndWhile
\State $\mathcal{P}(v) = i$
\EndFor
\end{algorithmic}
\caption{Initial traversal of the $\alpha$-neighborhood of $\sigma$}\label{alg:first-traversal}
\end{algorithm}

When \cref{alg:first-traversal} returns, we have added all vertices in all $\alpha$-neighborhoods of the vertices of $\sigma$ to a heap.
The next step is to pull from that heap, in vertex order, until we find the first vertex that is in all neighborhoods of the vertices of $\sigma$.
We describe this in \cref{alg:find-apparent-pair}.

\begin{algorithm}
\begin{algorithmic}
\While{$\mathcal{V}_{L1} \neq \emptyset$}
\State $w\gets \mathcal{V}_{L1}.\text{pop}$
\If{$w\in\mathcal{V}$}
\State $\mathcal{V}(w)\gets \mathcal{V}(w)+1$
\Else
\State $\mathcal{V}(w)\gets 1$
\EndIf
\If{$\mathcal{V}(w)=|\sigma|$}
\State \textbf{break} out of the while loop and return control
\EndIf
\EndWhile
\end{algorithmic}
\caption{Finding the first vertex in the intersection of all neighborhoods.}\label{alg:find-apparent-pair}
\end{algorithm}

After \cref{alg:find-apparent-pair} finishes, either $\mathcal{V}_{L1}$ is empty, or $w$ is the vertex that could produce an apparent pair.
We can check for an apparent pair at this stage, and skip all the subsequent steps when we find one.

If we do want to continue past the apparent pair check, we will keep on moving vertices from $\mathcal{V}_{L1}$ to $\mathcal{V}$, tracking whether the count in $\mathcal{V}$ is large enough to indicate an element of the intersection.
Each time we do find something in the intersection, we can add the cofacet to our output, and remove the entry in the vertex counting map.
We describe this in \cref{alg:immediate-cofacets}.

\begin{algorithm}
\begin{algorithmic}
\Require $w$, $\mathcal{V}_{L1}$, $\mathcal{V}$, $\mathcal{P}$ and $\mathcal{N}$ as described above and corresponding to the end state of \cref{alg:find-apparent-pair}.
\State $\mathcal{O} \gets \emptyset$
\If{$|\mathcal{V}(w)|=|\sigma|$}
\State $\mathcal{V}.\text{remove}(w)$
\State $\mathcal{O}.\text{append}(w)$
\EndIf
\While{$\mathcal{V}_{L1}\neq\emptyset$} 
\If{$|\mathcal{V}|(w)=|\sigma|$}
\State $\mathcal{V}.\text{remove}(w)$
\State $\mathcal{O}.\text{append}(w)$
\EndIf
\State $w\gets \mathcal{V}_{L1}.\text{pop}$
\If{$w\in\mathcal{V}$}
\State $\mathcal{V}(w)\gets \mathcal{V}(w)+1$
\Else
\State $\mathcal{V}(w)\gets 1$
\EndIf
\EndWhile
\end{algorithmic}
\caption{Generating all the cofacets of diameter $\alpha$. This algorithm benefits greatly from generating an \texttt{Iterator} and yielding control whenever $\mathcal{O}$ receives an additional element.}\label{alg:immediate-cofacets}
\end{algorithm}

\cref{alg:immediate-cofacets} can also be written to instead first update all counts in $\mathcal{V}$, emptying out $\mathcal{V}_{L1}$ entirely, and then picking out entries in $\mathcal{V}$ of full size to add the results to $\mathcal{O}$.

At the end of \cref{alg:immediate-cofacets}, all indexes in $\mathcal{P}$ point to vertices that are not in the $\alpha$-neighborhood of $\sigma$ and $\mathcal{O}$ has received all cofacets of diameter $\alpha$.
The final phase of the algorithm continues to advance whichever index in $\mathcal{P}$ that points to a closest neighbor, updating $\mathcal{V}$ and outputting whenever the new count in $\mathcal{V}$ is equal to the simplex size.
We describe this in \cref{alg:remaining-cofacets}.

\begin{algorithm}
\begin{algorithmic}
\Require $w$, $\mathcal{V}$, $\mathcal{P}$ and $\mathcal{N}$ as described above and corresponding to the end state of \cref{alg:immediate-cofacets}.
\Loop
\State $v \gets \arg\min_{v\in\sigma} d\left(v, \mathcal{N}(v)(\mathcal{P}(v)\right)$
\State $w \gets \mathcal{N}(v)(\mathcal{P}(v))$
\If{$w\in\mathcal{V}$}
\State $\mathcal{V}(w)\gets \mathcal{V}(w)+1$
\Else
\State $\mathcal{V}(w)\gets 1$
\EndIf
\If{$|\mathcal{V}(w)=|\sigma|$}
\State $\mathcal{V}.\text{remove}(w)$
\State $\mathcal{O}.\text{append}(w)$
\EndIf
\State $\mathcal{P}(v) \gets \mathcal{P}(v)+1$
\If{$\mathcal{P}(v)\geq |\mathcal{N}(v)|$}
\State \textbf{break} out of the loop
\EndIf
\EndLoop
\end{algorithmic}
\caption{Generating all the cofacets of diameter greater than $\alpha$.}\label{alg:remaining-cofacets}
\end{algorithm}

When any one of the indexes stored in $\mathcal{P}$ is larger than the list that it is indexing into, we know that there can be no additional vertices under consideration other than the ones that are stored in $\mathcal{P}$.
At this stage, we could either continue running \cref{alg:remaining-cofacets}, ignoring any neighbor lists that have been exhausted and output vertices that do generate cofacets in filtration order, or we could go directly to the metric space representation and compute the diameters of the remaining potential cofacets directly, sort those, and output them.

\section{Experiments}

\subsection{Simplex Stream generation}

To test the performance of the filtration order approach to computing simplex streams, we draw on previous work on benchmarking TDA packages \cite{RJ-2021-033}.
Based on their work, we are comparing computation times between our implementation of the approach used in Ripser, with our implementation of the method introduced in \cref{sec:simplex-stream}.
We compute and traverse the Vietoris-Rips complex up to simplices of dimension 4 (for $25$ or $50$ data points) or up to simplex dimension 3 (for larger data sets) of point clouds sampled uniformly at random in the unit cube.

The resulting timings can be seen in \cref{fig:traversal-timings}.
Simplex dimensions 0 and 1 are handled by listing vertices and pairs of vertices directly, which gives the Ripser approach an edge with less setup work, but with increasing problem sizes that edge quickly goes away.
From simplex dimension 2 and up, we see a consistent advantage to the in-order generation of about a factor $10\times$ across the board.
There is no clear reason to believe that there is an asymptotic difference in complexity: slopes in these logarithmic-scaled graphs look very similar.

\begin{figure}
    \includegraphics[width=1.25\linewidth]{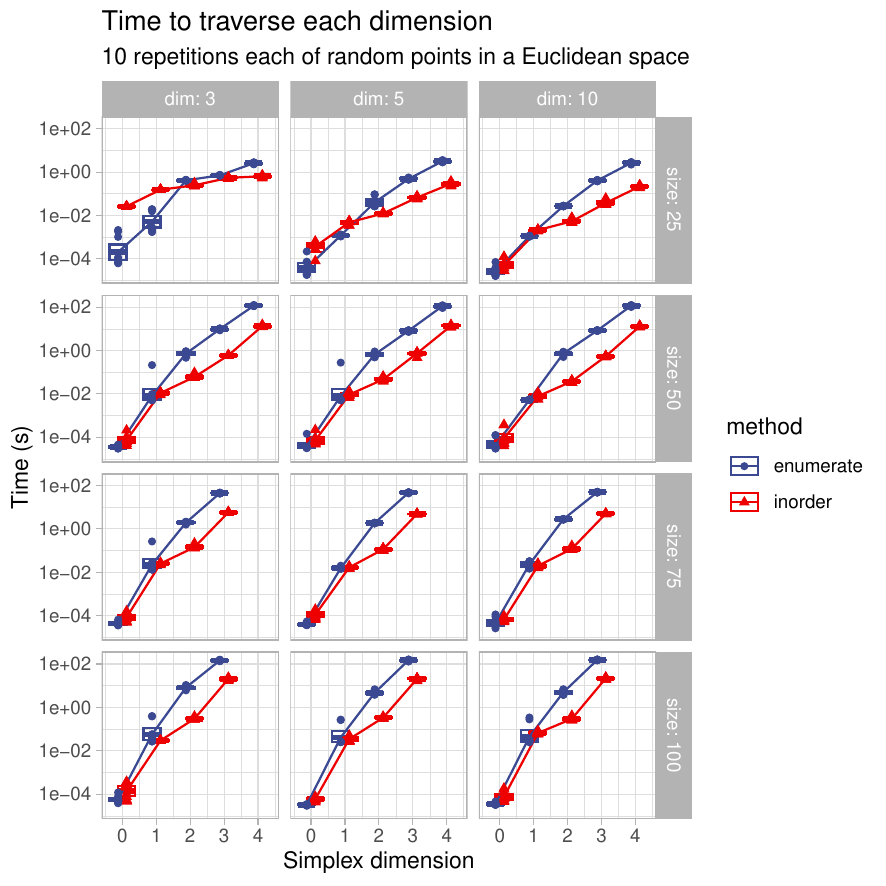}
    \caption{Time required to traverse each skeleton layer from 0 to 4 of a Vietoris-Rips complex on random points in $\mathbb{R}^d$.
    Mean and a non-parametric bootstrap confidence interval for the mean is displayed in the box, while each observation is a point: blue/left/round for the enumeration method used in Ripser and red/right/triangle for the method we propose here.
    Dimensions 0 and 1 are handled by directly writing vertex lists, or generating and sorting edges, while higher dimensions use the two different algorithms.
    Hence, the differences that can be seen in speeds for dimensions 0 and 1 are mainly due to a more elaborate setup for the inorder method.}
    \label{fig:traversal-timings}
\end{figure}

\subsection{Apparent Pair detection}

To measure the degree to which this approach improves on the detection of apparent pairs, we used the ScalaCheck library to generate random point clouds in $\mathbb{R}^5$ with between 25 and 250 points.
Limiting the neighborhood searches to the minimum enclosing radius, we computer Vietoris-Rips cofacets. For each point cloud, we picked a random simplex and iterated through its cofacets, measuring how many vertices $w$ were seen in the initiation phase.
Recall that at the end of this phase, as described in \cref{alg:first-traversal} and \cref{alg:find-apparent-pair}, if there could be an apparent pair, we have the vertex that could produce it.
Comparing the number of visited vertices to the index number of the resulting vertex gives us a comparison between our method and the index-ordered methods in Ripser for the speed of finding an Apparent Pair.

We show a comparison of the number of vertices visited to find an Apparent Pair candidate in \cref{fig:vertex-visits}.

\begin{figure}
    \hspace{-3cm}
    \begin{minipage}{1.2\linewidth}
    \input{visitsplot}
    \input{visitratioplot}        
    \end{minipage}
    \caption{The number of vertices checked before a candidate for an Apparent Pair is found, compared between our method and the method used in Ripser.\newline%
    We compare the two methods with two ECDF plots: in each plot, the curve indicates the proportion of observations with value less than or equal to the $x$-coordinate. To the left, simultaneous plots of the number of vertices visited before finding an Apparent Pair candidate, we see that the filtration order methods find Apparent Pair candidates vastly much faster than the Ripser methods. To the right, we show the paired behavior: for the same choice of complex and simplex, we take the ratio of the number of candidates visited for our method, divided by the number for the method in Ripser. We note that the break-even point is at around 80-85\%: in the vast majority of cases, our method is faster, in some cases by several orders of magnitude.}
    \label{fig:vertex-visits}
\end{figure}

\subsection{Experiments in Ripser}
\label{sec:experiments-in-ripser}

In order to get a measure of how well these proposed methods fare in the Ripser code base, we changed the Ripser code base to include using the fast detection and the simplex stream generation using \cref{thm:three-cases}.

In \cref{tbl:ripser-time-mem}, we show execution time measurements and peak memory consumption for 5 benchmarking datasets from the original Ripser paper, and the Ripser repository.
Each measurement was done 10 times on a 

\begin{table}
\begin{tabular}{lrrrrrrr}
\toprule
 &  &  &  & \multicolumn{2}{c}{time} & \multicolumn{2}{c}{mem} \\ 
\cmidrule(lr){5-6} \cmidrule(lr){7-8}
 &  &  &  & ripser & inorder & ripser & inorder \\ 
\cmidrule(lr){5-5} \cmidrule(lr){6-6} \cmidrule(lr){7-7} \cmidrule(lr){8-8}
case & n & dim & thresh. & mean (sd) & mean (sd) & mean (sd) & mean (sd) \\ 
\midrule
o3\_1024 & 1\,024 & 4 & 2.000 & \textbf{66s} (11s) & 97s (8s) & 169.8 MB (35.9 kB) & \textbf{59.1 MB} (34.7 kB) \\ 
o3\_4096 & 4\,096 & 4 & 1.500 & \textbf{1\,302s} (138s) & 1\,916s (131s) & 2.3 GB (45.5 kB) & \textbf{722.6 MB} (48.5 kB) \\ 
CycloOctane & 6\,040 & 2 & 1.500 & \textbf{1\,304s} (104s) & 1\,395s (108s) & 8.9 GB (59.4 kB) & \textbf{7.1 GB} (67.1 kB) \\ 
rp2\_600 & 600 & 3 & 1.000 & \textbf{176s} (25s) & 233s (23s) & 186.5 MB (38.6 kB) & \textbf{173.2 MB} (30.8 kB) \\ 
sphere\_3\_192 & 192 & 3 & 1.974 & \textbf{32s} (5s) & 40s (4s) & 122.6 MB (44.7 kB) & \textbf{109 MB} (42.3 kB) \\ 
\bottomrule
\end{tabular}
    \caption{Mean (and standard deviation) of sets of 10 runs with each code base for 5 different benchmarking datasets used in the original Ripser paper.}
    \label{tbl:ripser-time-mem}
\end{table}

\section{Discussion}
\label{sec:discussion}

We have presented algorithms for generating cofacets in order to create a simplex sequence, to recognize apparent pairs or to create coboundaries of specific simplices.

By creating a simplex sequence with a stack of enumerators, the entirety of the simplex layer $\Sigma_{d-1}$ need no longer be created before $\Sigma_d$ is enumerated. Instead, we can create $\Sigma_d$ directly, with minimal memory residency. 

In our approach to recognizing apparent pairs, \cref{fig:vertex-visits} shows that we recognize an apparent pair far earlier than the existing approach in Ripser, which cuts down on the work needed before the simplex can be abandoned entirely.

\cref{tbl:ripser-time-mem} shows that these two improvements, when implemented in Ripser, contribute to an increase in computation time -- in our experiments ranging between $7\%$ and $50\%$, and to a consistent sometimes dramatic decrease in memory consumption -- ranging between $7\%$ and $65\%$ improvements. The improved memory consumption is clearer for memory intensive computations (in very small cases, the overhead of the extra book-keeping is noticeable), and clearer for higher homological dimension.

While this paper has been written in the context of Vietoris-Rips complexes, we believe that generalizations to clique complexes of other filtered graphs, or to {\v C}ech style complexes should be possible.

We further believe that these results stake out a path for experimenting in the Ripser code base with representing each simplex coboundary by an enumerator, and only generating the next simplex only when cancellations happen. A linear combination of $k$ different simplex coboundaries would then be represented by $k$ different enumerator objects, where the enumerator with the smallest current simplex would be carrying the pivot. This approach is directly inspired by \cite{aggarwal_dory_2024-1} and by conversations with Florian Roll.

Another direction suggested by this research is to follow the direction suggested in \cite{vejdemo2011interleaved}: generate each simplex directly before reducing it, and keep \emph{nothing} in memory until it is needed. Instead of ever needing to represent a full layer of simplices (or even simplices not in apparent pairs), we would only need to represent simplices whose coboundaries have already been reduced to a non-zero coboundary basis element.

\bibliography{main.bib}

\appendix
\section{Worked example}
\label{app:worked-example}

Consider $8$ points distributed evenly on the unit circle, and their Vietoris-Rips complex.
They have distance matrix:

\[
\text{dist} = \begin{pmatrix}
  0.00 & 0.77 & 1.41 & 1.85 & 2.00 & 1.85 & 1.41 & 0.77 \\
  0.77 & 0.00 & 0.77 & 1.41 & 1.85 & 2.00 & 1.85 & 1.41 \\
  1.41 & 0.77 & 0.00 & 0.77 & 1.41 & 1.85 & 2.00 & 1.85 \\
  1.85 & 1.41 & 0.77 & 0.00 & 0.77 & 1.41 & 1.85 & 2.00 \\
  2.00 & 1.85 & 1.41 & 0.77 & 0.00 & 0.77 & 1.41 & 1.85 \\
  1.85 & 2.00 & 1.85 & 1.41 & 0.77 & 0.00 & 0.77 & 1.41 \\
  1.41 & 1.85 & 2.00 & 1.85 & 1.41 & 0.77 & 0.00 & 0.77 \\
  0.77 & 1.41 & 1.85 & 2.00 & 1.85 & 1.41 & 0.77 & 0.00 \\
  \end{pmatrix}
\]

The Dory-style neighborhood lists $\mathcal{N}$ are:

\begin{align*}
  \begin{pmatrix}0.00 \\ 0\end{pmatrix}\quad\begin{pmatrix}0.77 \\ 1\end{pmatrix}\quad\begin{pmatrix}0.77 \\ 7\end{pmatrix}\quad\begin{pmatrix}1.41 \\ 2\end{pmatrix}\quad\begin{pmatrix}1.41 \\ 6\end{pmatrix}\quad\begin{pmatrix}1.85 \\ 3\end{pmatrix}\quad\begin{pmatrix}1.85 \\ 5\end{pmatrix}\quad\begin{pmatrix}2.00 \\ 4\end{pmatrix} \\
  \begin{pmatrix}0.00 \\ 1\end{pmatrix}\quad\begin{pmatrix}0.77 \\ 0\end{pmatrix}\quad\begin{pmatrix}0.77 \\ 2\end{pmatrix}\quad\begin{pmatrix}1.41 \\ 3\end{pmatrix}\quad\begin{pmatrix}1.41 \\ 7\end{pmatrix}\quad\begin{pmatrix}1.85 \\ 4\end{pmatrix}\quad\begin{pmatrix}1.85 \\ 6\end{pmatrix}\quad\begin{pmatrix}2.00 \\ 5\end{pmatrix} \\
  \begin{pmatrix}0.00 \\ 2\end{pmatrix}\quad\begin{pmatrix}0.77 \\ 1\end{pmatrix}\quad\begin{pmatrix}0.77 \\ 3\end{pmatrix}\quad\begin{pmatrix}1.41 \\ 0\end{pmatrix}\quad\begin{pmatrix}1.41 \\ 4\end{pmatrix}\quad\begin{pmatrix}1.85 \\ 5\end{pmatrix}\quad\begin{pmatrix}1.85 \\ 7\end{pmatrix}\quad\begin{pmatrix}2.00 \\ 6\end{pmatrix} \\
  \begin{pmatrix}0.00 \\ 3\end{pmatrix}\quad\begin{pmatrix}0.77 \\ 4\end{pmatrix}\quad\begin{pmatrix}0.77 \\ 2\end{pmatrix}\quad\begin{pmatrix}1.41 \\ 1\end{pmatrix}\quad\begin{pmatrix}1.41 \\ 5\end{pmatrix}\quad\begin{pmatrix}1.85 \\ 6\end{pmatrix}\quad\begin{pmatrix}1.85 \\ 0\end{pmatrix}\quad\begin{pmatrix}2.00 \\ 7\end{pmatrix} \\
  \begin{pmatrix}0.00 \\ 4\end{pmatrix}\quad\begin{pmatrix}0.77 \\ 3\end{pmatrix}\quad\begin{pmatrix}0.77 \\ 5\end{pmatrix}\quad\begin{pmatrix}1.41 \\ 2\end{pmatrix}\quad\begin{pmatrix}1.41 \\ 6\end{pmatrix}\quad\begin{pmatrix}1.85 \\ 1\end{pmatrix}\quad\begin{pmatrix}1.85 \\ 7\end{pmatrix}\quad\begin{pmatrix}2.00 \\ 0\end{pmatrix} \\
  \begin{pmatrix}0.00 \\ 5\end{pmatrix}\quad\begin{pmatrix}0.77 \\ 6\end{pmatrix}\quad\begin{pmatrix}0.77 \\ 4\end{pmatrix}\quad\begin{pmatrix}1.41 \\ 3\end{pmatrix}\quad\begin{pmatrix}1.41 \\ 7\end{pmatrix}\quad\begin{pmatrix}1.85 \\ 0\end{pmatrix}\quad\begin{pmatrix}1.85 \\ 2\end{pmatrix}\quad\begin{pmatrix}2.00 \\ 1\end{pmatrix} \\
  \begin{pmatrix}0.00 \\ 6\end{pmatrix}\quad\begin{pmatrix}0.77 \\ 5\end{pmatrix}\quad\begin{pmatrix}0.77 \\ 7\end{pmatrix}\quad\begin{pmatrix}1.41 \\ 0\end{pmatrix}\quad\begin{pmatrix}1.41 \\ 4\end{pmatrix}\quad\begin{pmatrix}1.85 \\ 3\end{pmatrix}\quad\begin{pmatrix}1.85 \\ 1\end{pmatrix}\quad\begin{pmatrix}2.00 \\ 2\end{pmatrix} \\
  \begin{pmatrix}0.00 \\ 7\end{pmatrix}\quad\begin{pmatrix}0.77 \\ 6\end{pmatrix}\quad\begin{pmatrix}0.77 \\ 0\end{pmatrix}\quad\begin{pmatrix}1.41 \\ 5\end{pmatrix}\quad\begin{pmatrix}1.41 \\ 1\end{pmatrix}\quad\begin{pmatrix}1.85 \\ 2\end{pmatrix}\quad\begin{pmatrix}1.85 \\ 4\end{pmatrix}\quad\begin{pmatrix}2.00 \\ 3\end{pmatrix} \\
  \end{align*}

\subsection{Simplex Stream generation}

\paragraph*{Dim 0}
Just the vertices.

\[
\mathcal{O} = [[0],[1],[2],[3],[4],[5],[6],[7]]
\]

\paragraph*{Dim 1}
Edges, in sorted order.

\begin{multline*}
\mathcal{O} = [[0,1],[0,7],[1,2],[2,3],[3,4],[4,5],[5,6],[6,7],\\
[0,2],[0,6],[1,3],[2,4],[3,5],[4,6],[5,7],\\
[0,3],[0,5],[1,4],[2,5],[3,6],[4,7],\\
[0,4],[1,5],[2,6],[3,7]] \\
\end{multline*}

\paragraph*{Dim 2}

These are the vertices to be considered in each case: for vertices in Case 1 we only need to check that they are non-expanding; in Case 2 that all longest edges connect to the first vertex; in Case 3 that no new edges are full length.

\begin{tabular}{l|l|l|l|l|l}
$\sigma$ & Case 1 & Case 2 & Case 3 & $\mathcal{O}$ & Comment \\
\hline
$[0,1]$ & & & 2,3,4,5,6,7 & \\
$[0,7]$ & & 1,2,3,4,5,6 & & \\
$[1,2]$ & 0 & & 3,4,5,6,7 & \\
$[2,3]$ & 0,1 & & 4,5,6,7 & \\
$[3,4]$ & 0,1,2 & & 5,6,7 & \\
$[4,5]$ & 0,1,2,3 & & 6,7 & \\
$[5,6]$ & 0,1,2,3,4 & & 7 & \\
$[6,7]$ & 0,1,2,3,4,5 & & & \\
$[0,2]$ & & 1 & 3,4,5,6,7 & 1 \\
$[0,6]$ & & 1,2,3,4,5 & 7 & 7 \\
$[1,3]$ & 0 & 2 & 4,5,6,7 & 2 \\
$[1,7]$ & 0 & 2,3,4,5,6 & & 0 \\
$[2,4]$ & 0,1 & 3 & 5,6,7 & 3 \\
$[3,5]$ & 0,1,2 & 4 & 6,7 & 4 \\
$[4,6]$ & 0,1,2,3 & 5 & 7 & 5 \\
$[5,7]$ & 0,1,2,3,4 & 6 & & 6 \\
$[0,3]$ & & 1,2 & 4,5,6,7 & 1,2 & a) \\
$[0,5]$ & & 1,2,3,4 & 6,7 & 3,6,7 \\
$[1,4]$ & 0 & 2,3 & 5,6,7 & 2,3,7 & b) \\
$[1,6]$ & 0 & 2,3,4,5 & 7 & 0,4,7 \\
$[2,5]$ & 0,1 & 3,4 & 6,7 & 0,3,4 \\
$[2,7]$ & 0,1 & 3,4,5,6 & & 0,1,5 \\
$[3,6]$ & 0,1,2 & 4,5 & 7 & 0,1,4,5 \\
$[4,7]$ & 0,1,2,3 & 5,6 & & 1,2,5,6 \\
$[0,4]$ & & 1,2,3 & 5,6,7 & 1,2,3,5,6,7 \\
$[1,5]$ & 0 & 2,3,4 & 6,7 & 0,2,3,4,6,7 \\
$[2,6]$ & 0,1 & 3,4,5 & 7 & 0,1,3,4,5,7 \\
$[3,7]$ & 0,1,2 & 4,5,6 & & 0,1,2,4,5,6
\end{tabular}

\begin{enumerate}[a)]
\item $\alpha=1.85$ so both 5 and 6 are non-expanding; only 5 makes $[0,3]$ the unique longest edge in the result.
\item Since $d(0,4)=2.00$, the option $0$ is not non-expanding. The option $6$ brings in too many longest edges, while $7$ keeps $[1,4]$ the unique longest edge.
\end{enumerate}

\paragraph*{Dim 3}

In this table, we highlight options that can be skipped completely due to the structure of $\sigma$; as in the case of $[0,1,7]$, where $d(0,1)$ is already less than $d(0,7)$, so in no cofacet of $[0,1,7]$ is $[0,1]$ the \emph{unique} longest edge; or in the case of $[1,4,7]$ where both $[1,4]$ and $[4,7]$ are longest edges, so it in no cofacet of $[1,4,7]$ do all longest edges connect to 1.

\begin{minipage}[t]{0.4\textwidth}
\begin{tabular}{l|l|l|l|l|l}
$\sigma$ & Case 1 & Case 2 & Case 3 & $\mathcal{O}$ \\
\hline
$[0,1,2]$ & & & &  \\
$[0,6,7]$ & & 1,2,3,4,5 & & \\
$[1,2,3]$ & 0 & & & \\
$[0,1,7]$ & & & \emph{2,3,4,5,6} & \\
$[2,3,4]$ & 0,1 & & & \\
$[3,4,5]$ & 0,1,2 & & & \\
$[4,5,6]$ & 0,1,2,3 & & & \\
$[5,6,7]$ & 0,1,2,3,4 & & & \\
$[0,1,3]$ & & & \emph{2} & \\
$[0,2,3]$ & & 1 & & 1 \\
$[0,2,5]$ & & 1 & \emph{3,4} & \\
$[0,3,5]$ & & 1,2 & 4 & \\
$[0,5,6]$ & & 1,2,3,4 & & \\
$[0,5,7]$ & & 1,2,3,4 & 6 & 6 \\
$[1,2,4]$ & 0 & & \emph{3} & \\
$[1,3,4]$ & 0 & 2 & & 2 \\
$[1,4,7]$ & 0 & \emph{2,3} & \emph{5,6} & \\
$[0,1,6]$ & & & \emph{2,3,4,5} & \\
$[1,4,6]$ & 0 & 2,3 & 5 & \\
$[1,6,7]$ & 0 & 2,3,4,5 & & 0 \\
$[0,2,5]$ & & \emph{1} & \emph{3,4} & \\
$[2,3,5]$ & 0,1 & & \emph{4} & \\
$[2,4,5]$ & 0,1 & 3 & & 3 \\
$[0,2,7]$ & & \emph{1} & \emph{3,4,5,6} & \\
$[1,2,7]$ & 0 & & \emph{3,4,5,6} & 0 \\
$[2,5,7]$ & 0,1 & \emph{3,4} & \emph{6} & 0 \\
$[0,3,6]$ & & \emph{1,2} & \emph{4,5} & \\
$[1,3,6]$ & 0 & \emph{2} & \emph{4,5} & 0 \\
$[3,4,6]$ & 0,1,2 & & \emph{5} & 1 \\
$[3,5,6]$ & 0,1,2 & 4 & & 4 \\
$[1,4,7]$ & 0 & \emph{2,3} & \emph{5,6} & \\
$[2,4,7]$ & 0,1 & \emph{3} & \emph{5,6} & 1 \\
$[4,5,7]$ & 0,1,2,3 & & \emph{6} & 2 \\
$[4,6,7]$ & 0,1,2,3 & 5 & & 1,5 \\
\end{tabular}
\end{minipage}
\hfill
\begin{minipage}[t]{0.4\textwidth}
\begin{tabular}{l|l|l|l|l|l}
$\sigma$ & Case 1 & Case 2 & Case 3 & $\mathcal{O}$ \\
\hline
$[0,1,4]$ & & & \emph{2,3} &  \\
$[0,2,4]$ & & 1 & \emph{3} & 1 \\
$[0,3,4]$ & & 1,2 & & 1,2 \\
$[0,4,5]$ & & 1,2,3 &  & 2,3 \\
$[0,4,6]$ & & 1,2,3 & 5 & 1,3,5 \\
$[0,4,7]$ & & 1,2,3 & 5,6 & 1,2,5,6 \\
$[0,1,5]$ & & & \emph{2,3,4} &  \\
$[1,2,5]$ & 0 & & \emph{3,4} & 0 \\
$[1,3,5]$ & 0 & 2 & \emph{4} & 0,2 \\
$[1,4,5]$ & 0 & 2,3 & & 0,2,3 \\
$[1,5,6]$ & 0 & 2,3,4 & & 0,3,4 \\
$[1,5,7]$ & 0 & 2,3,4 & 6 & 0,2,6 \\
$[0,2,6]$ &  & \emph{3,4,5} & \emph{7} & \\
$[1,2,6]$ & 0 & & \emph{3,4,5} & 0 \\
$[2,3,6]$ & 0,1 & & \emph{4,5} & 0,1 \\
$[2,4,6]$ & 0,1 & 3 & \emph{5} & 0,1,3 \\
$[2,5,6]$ & 0,1 & 3,4 & & 0,1,3,4 \\
$[2,6,7]$ & 0,1 & 3,4,5 & & 0,1,4,5 \\
$[0,3,7]$ & & \emph{1,2} & \emph{4,5,6} & \\
$[1,3,7]$ & 0 & \emph{2} & \emph{4,5,6} & 0 \\
$[2,3,7]$ & 0,1 & & \emph{4,5,6} & 0,1 \\
$[3,4,7]$ & 0,1,2 & & \emph{5,6} & 0,1,2 \\
$[3,5,7]$ & 0,1,2 & 4 & \emph{6} & 0,1,2,4 \\
$[3,6,7]$ & 0,1,2 & 4,5 & & 0,1,2,4,5 \\
\end{tabular}
\end{minipage}

\subsection{Apparent Pairs and Coboundary Generation}

Suppose we are looking for cofacets of $[0 1 3]$.
This simplex has diameter 1.85.

\paragraph*{\cref{alg:first-traversal}}
\begin{multline*}
v=0 \\
\mathcal{N}(v) = \begin{pmatrix}0.00\\0\end{pmatrix}\quad\begin{pmatrix}0.77 \\ 1\end{pmatrix}\quad\begin{pmatrix}0.77 \\ 7\end{pmatrix}\quad\begin{pmatrix}1.41 \\ 2\end{pmatrix}\quad\begin{pmatrix}1.41 \\ 6\end{pmatrix}\quad\begin{pmatrix}1.85 \\ 3\end{pmatrix}\quad\begin{pmatrix}1.85 \\ 5\end{pmatrix}\quad\begin{pmatrix}2.00 \\ 4\end{pmatrix}
\\
i = 5 \qquad\mathcal{V}_{L1} = [0,1,2,6,7] \qquad\mathcal{P}=[0\to 5]
\end{multline*}

\begin{multline*}
v=1 \\
\mathcal{N}(v) =   \begin{pmatrix}0.00 \\ 1\end{pmatrix}\quad\begin{pmatrix}0.77 \\ 0\end{pmatrix}\quad\begin{pmatrix}0.77 \\ 2\end{pmatrix}\quad\begin{pmatrix}1.41 \\ 3\end{pmatrix}\quad\begin{pmatrix}1.41 \\ 7\end{pmatrix}\quad\begin{pmatrix}1.85 \\ 4\end{pmatrix}\quad\begin{pmatrix}1.85 \\ 6\end{pmatrix}\quad\begin{pmatrix}2.00 \\ 5\end{pmatrix} \\
\\
i = 5 \qquad\mathcal{V}_{L1} = [0,0,1,1,2,2,3,6,7,7] \qquad\mathcal{P}=[0\to 5, 1\to 5]
\end{multline*}

\begin{multline*}
v=3 \\
\mathcal{N}(v) =   \begin{pmatrix}0.00 \\ 3\end{pmatrix}\quad\begin{pmatrix}0.77 \\ 4\end{pmatrix}\quad\begin{pmatrix}0.77 \\ 2\end{pmatrix}\quad\begin{pmatrix}1.41 \\ 1\end{pmatrix}\quad\begin{pmatrix}1.41 \\ 5\end{pmatrix}\quad\begin{pmatrix}1.85 \\ 6\end{pmatrix}\quad\begin{pmatrix}1.85 \\ 0\end{pmatrix}\quad\begin{pmatrix}2.00 \\ 7\end{pmatrix} \\
\\
i = 5 \qquad\mathcal{V}_{L1} = [0,0,1,1,1,2,2,2,3,3,4,5,6,7,7] \qquad\mathcal{P}=[0\to 5, 1\to 5, 3\to 5]
\end{multline*}

\paragraph*{\cref{alg:find-apparent-pair}}

\begin{tabular}{l|l|l}
$w$ & $\mathcal{V}_{L1}$ & $\mathcal{V}$ \\
\hline
$0$ & $[0,1,1,1,2,2,2,3,3,4,5,6,7,7]$ & $[0\to 1]$ \\
$0$ & $[1,1,1,2,2,2,3,3,4,5,6,7,7]$ & $[0\to 2]$ \\
$1$ & $[1,1,2,2,2,3,3,4,5,6,7,7]$ & $[0\to 2, 1\to 1]$ \\
$1$ & $[1,2,2,2,3,3,4,5,6,7,7]$ & $[0\to 2, 1\to 2]$ \\
$1$ & $[2,2,2,3,3,4,5,6,7,7]$ & $[0\to 2, 1\to 3]$ \\
$2$ & $[2,2,3,3,4,5,6,7,7]$ & $[0\to 2, 1\to 3, 2\to 1]$ \\
$2$ & $[2,3,3,4,5,6,7,7]$ & $[0\to 2, 1\to 3, 2\to 2]$ \\
$2$ & $[3,3,4,5,6,7,7]$ & $[0\to 2, 1\to 3, 2\to 3]$ \\
\end{tabular}

Here, we break out of the loop, and return $2$ as a candidate apparent pair vertex, for the coface $[0 1 2 3]$.
The preceding $1$ that also counted to 3 is ignored since it is one of the vertices of our query simplex.

\paragraph*{\cref{alg:immediate-cofacets}}
\begin{tabular}{l|l|l|l}
$w$ & $\mathcal{V}_{L1}$ & $\mathcal{V}$ & $\mathcal{O}$ \\
\hline
$3$ & $[3,4,5,6,7,7]$ & $[0\to 2, 1\to 3, 2\to 3, 3\to 1]$ & [2] \\
$3$ & $[4,5,6,7,7]$ & $[0\to 2, 1\to 3, 2\to 3, 3\to 2]$ & [2] \\
$4$ & $[5,6,7,7]$ & $[0\to 2, 1\to 3, 2\to 3, 3\to 2, 4\to 1]$ & [2] \\
$5$ & $[6,7,7]$ & $[0\to 2, 1\to 3, 2\to 3, 3\to 2, 4\to 1, 5\to 1]$ & [2] \\
$6$ & $[7,7]$ & $[0\to 2, 1\to 3, 2\to 3, 3\to 2, 4\to 1, 5\to 1,6\to 1]$ & [2] \\
$7$ & $[7]$ & $[0\to 2, 1\to 3, 2\to 3, 3\to 2, 4\to 1, 5\to 1,6\to 1, 7\to 1]$ & [2] \\
$7$ & $[]$ & $[0\to 2, 1\to 3, 2\to 3, 3\to 2, 4\to 1, 5\to 1,6\to 1, 7\to 2]$ & [2] \\
\end{tabular}

\paragraph*{\cref{alg:remaining-cofacets}}
\begin{tabular}{l|l|l|l|l}
$\mathcal{P}$ & $\mathcal{N}(v)(\mathcal{P}(v))$ & $w$ & $\mathcal{V}$ & $\mathcal{O}$ \\
\hline
$[0\to 5, 1\to 5, 3\to 5]$ & $\begin{pmatrix}1.85\\3\end{pmatrix}\begin{pmatrix}1.85\\4\end{pmatrix}\begin{pmatrix}1.85\\6\end{pmatrix}$ & 3 & $[0\to 2, 3\to 3, 4\to 1, 5\to 1,6\to 1, 7\to 2]$ & $[2]$ \\
$[0\to 6, 1\to 5, 3\to 5]$ & $\begin{pmatrix}1.85\\5\end{pmatrix}\begin{pmatrix}1.85\\4\end{pmatrix}\begin{pmatrix}1.85\\6\end{pmatrix}$ & 4 & $[0\to 2, 4\to 2, 5\to 1,6\to 1, 7\to 2]$ & $[2]$ \\
$[0\to 6, 1\to 6, 3\to 5]$ & $\begin{pmatrix}1.85\\5\end{pmatrix}\begin{pmatrix}1.85\\6\end{pmatrix}\begin{pmatrix}1.85\\6\end{pmatrix}$ & 5 & $[0\to 2, 4\to 2, 5\to 2,6\to 1, 7\to 2]$ & $[2]$ \\
$[0\to 7, 1\to 6, 3\to 5]$ & $\begin{pmatrix}2.00\\4\end{pmatrix}\begin{pmatrix}1.85\\6\end{pmatrix}\begin{pmatrix}1.85\\6\end{pmatrix}$ & 6 & $[0\to 2, 4\to 2, 5\to 2,6\to 2, 7\to 2]$ & $[2]$ \\
$[0\to 7, 1\to 7, 3\to 5]$ & $\begin{pmatrix}2.00\\4\end{pmatrix}\begin{pmatrix}2.00\\5\end{pmatrix}\begin{pmatrix}1.85\\6\end{pmatrix}$ & 6 & $[0\to 2, 4\to 2, 5\to 2,6\to 3, 7\to 2]$ & $[2,6]$ \\
$[0\to 7, 1\to 7, 3\to 6]$ & $\begin{pmatrix}2.00\\4\end{pmatrix}\begin{pmatrix}2.00\\5\end{pmatrix}\begin{pmatrix}1.85\\0\end{pmatrix}$ & 0 & $[0\to 3, 4\to 2, 5\to 2, 7\to 2]$ & $[2,6]$ \\
$[0\to 7, 1\to 7, 3\to 7]$ & $\begin{pmatrix}2.00\\4\end{pmatrix}\begin{pmatrix}2.00\\5\end{pmatrix}\begin{pmatrix}2.00\\7\end{pmatrix}$ & 4 & $[4\to 3, 5\to 2, 7\to 2]$ & $[2,6,4]$ \\
$[0\to 8, 1\to 7, 3\to 7]$ & $\phantom{\begin{pmatrix}2.00\\4\end{pmatrix}}\begin{pmatrix}2.00\\5\end{pmatrix}\begin{pmatrix}2.00\\7\end{pmatrix}$ & 5 & $[5\to 3, 7\to 2]$ & $[2,6,4,5]$ \\
$[0\to 8, 1\to 8, 3\to 7]$ & $\phantom{\begin{pmatrix}2.00\\4\end{pmatrix}}\phantom{\begin{pmatrix}2.00\\5\end{pmatrix}}\begin{pmatrix}2.00\\7\end{pmatrix}$ & 7 & $[7\to 3]$ & $[2,6,4,5,7]$ \\
$[0\to 8, 1\to 8, 3\to 8]$ &  &  & $[]$ & $[2,6,4,5,7]$ \\
\end{tabular}

Each row in this last table shows: the state of $\mathcal{P}$ at the beginning, the three neighbors of the vertices pointed to by the indices in $\mathcal{P}$, the resulting choice for $w$, the table $\mathcal{V}$ after updating and finally the output queue after catching and outputting any vertex not in the simplex $[0 1 3]$ that accumulates a total count of 3.

The last row described in the table documents the state of all the data structures involved after the algorithm returns.

\end{document}